\documentclass[a4paper, 11pt]{amsart}
\usepackage{amsmath,amsfonts,amssymb}
\usepackage{graphicx,amsmath, amscd}
\usepackage{color,comment}
\usepackage[T1]{fontenc}
\usepackage[latin1]{inputenc}
\usepackage[lmargin=3cm,rmargin=3cm,tmargin=4cm,bmargin=4cm]{geometry}
\usepackage[toc,page]{appendix}
\usepackage{tikz}
\usepackage{url}

\newtheorem{theorem}{Theorem}[section]
\newtheorem{lemma}[theorem]{Lemma}
\newtheorem{corollary}[theorem]{Corollary}
\newtheorem{proposition}[theorem]{Proposition}

\newtheorem{defi/prop}[theorem]{Definition/Proposition}

\newcommand{\N}{\mathbb{N}}

\newcommand{\R}{\mathbb{R}}

\newcommand{\E}{\mathbb{E}}
\renewcommand{\P}{\mathbb{P}}

\renewcommand{\leq}{\leqslant}
\renewcommand{\geq}{\geqslant}

\newcommand{\st}{\  : \ }

\DeclareMathOperator{\Tr}{Tr}

\newcommand{\ketbra}[2]{| #1 \rangle\!\langle #2 |}
\newcommand{\bra}[1]{\langle #1 |}
\newcommand{\ket}[1]{| #1 \rangle}

\author{C.E.~Gonz\'{a}lez-Guill\'{e}n}
\author{C.~Lancien}
\author{C.~Palazuelos}
\author{I.~Villanueva}

\address{\textbf{Carlos Gonz\'{a}lez-Guill\'{e}n:} Departamento de Matem\'{a}ticas del \'{A}rea Industrial, E.T.S.I. Industriales, Universidad Polit\'{e}cnica de Madrid, 28006 Madrid, Spain \& Instituto de Matem\'{a}tica Interdisciplinar, Universidad Complutense de Madrid, 28040 Madrid, Spain.}
\email{carlos.gguillen@upm.es}

\address{\textbf{C\'{e}cilia Lancien:} Departamento de An\'{a}lisis Matem\'{a}tico, Universidad Complutense de Madrid, 28040 Madrid, Spain.} \email{cecilia.lancien@free.fr}

\address{\textbf{Carlos Palazuelos:} Departamento de An\'{a}lisis Matem\'{a}tico, Universidad Complutense de Madrid, 28040 Madrid, Spain \& Instituto de Ciencias Matem\'{a}ticas, 28049 Madrid, Spain.}
\email{carlospalazuelos@ucm.es}

\address{\textbf{Ignacio Villanueva:} Departamento de An\'{a}lisis Matem\'{a}tico, Universidad Complutense de Madrid, 28040 Madrid, Spain \& Instituto de Matem\'{a}tica Interdisciplinar, Universidad Complutense de Madrid, 28040 Madrid, Spain.}
\email{ignaciov@ucm.es}

\title{Random quantum correlations are generically non-classical}

\begin{document}

\maketitle

\begin{abstract}
It is now a well-known fact that the correlations arising from local dichotomic measurements on an entangled quantum state may exhibit intrinsically non-classical features. In this paper we delve into a comprehensive study of \emph{random} instances of such bipartite correlations. The main question we are interested in is: given a quantum correlation, taken at random, how likely is it that it is truly non-explainable by a classical model? We show that, under very general assumptions on the considered distribution, a random correlation which lies on the border of the quantum set is with high probability outside the classical set. What is more, we are able to provide the Bell inequality certifying this fact. On the technical side, our results follow from (i) estimating precisely the ``quantum norm'' of a random matrix, and (ii) lower bounding sharply enough its ``classical norm'', hence proving a gap between the two. Along the way, we need a non-trivial upper bound on the $\infty{\rightarrow}1$ norm of a random orthogonal matrix, which might be of independent interest.
\end{abstract}

\section{Introduction}
The existence of quantum bipartite correlations which cannot be explained in a local realistic universe is one of the main features of quantum mechanics, both from a theoretical and an applied point of view. This phenomenon, known as \textit{quantum non-locality}, dates back to the mid 20th century (\cite{EPR}, \cite{Bell}) and it was first seen as a purely theoretical issue which could potentially lead to experimental verifications of the non-locality of nature. After a great effort, the recent experimental verification of quantum non-locality \cite{Hensen15} is indeed the strongest evidence we have that nature does not obey the classical laws. However, the main reason why quantum non-locality has become a central topic in quantum information theory is its great relevance in a variety of applications such as cryptography \cite{Aci+,AGM}, communication complexity \cite{BCMdW} or random number generators \cite{Pironio10}.

Most of the research in the understanding of quantum non-locality has focused in particular cases of quantum correlations and their dual objects, Bell inequalities. However, so far, we do not know much about the {\em generic case}. That is, if we consider a {\em random} correlation following a given probability distribution, can we say something about its probability of being quantum or classical, or something about how far it is, typically, from being any of them?

One of the first steps in this direction was given in  \cite{Amb+}, where the authors study the dual question, that is: how likely is it for a random (in a certain sense) Bell inequality to attain a strictly higher value on quantum correlations than on classical correlations? Later, in \cite{GJPV}, some of the authors of this note initiated the study of random correlations in the particular case where these correlations arise as the product of two rectangular normalized Gaussian matrices, a setting motivated by a well known  result of Tsirelson that we explain below. These results, among others,  are summarized in greater depth in the survey paper \cite{Pal}.

In this note, we study a very comprehensive family of random correlations, namely those which are bi-orthogonally invariant (this means that their probability distribution does not change when we multiply them by an orthogonal matrix, either from the right or from the left). The important families of Haar distributed orthogonal matrices and Gaussian matrices are of this kind, as well as the products of Gaussian rectangular matrices mentioned above.

Before going any further, let us recall the precise definitions of what we mean when we talk about classical vs quantum correlations.

Let $\rho$ be a bipartite (entangled) quantum state, on some tensor product complex Hilbert space $\mathrm{H}\otimes\mathrm{K}$, shared by two local observers. Assume that each of them can perform a binary-outcome measurement on his part of $\rho$, which he can choose amongst a set of $n$. What we are interested in is what is usually referred to as the \textit{quantum correlation matrix} $\tau$ arising from this scenario, which is the $n\times n$ real matrix defined by: for each $1\leq i,j\leq n$, $\tau_{i,j}=2\pi_{i,j}-1$, where $\pi_{i,j}$ is the probability that both observers obtain the same outcome, given that they have performed measurements $i$ on $\mathrm{H}$ and $j$ on $\mathrm{K}$, respectively. Such quantum correlation matrices can actually be characterized in an alternative way, which has the advantage of being mathematically very simple. Indeed, by a famous result of Tsirelson \cite{Tsi}, we know that $\tau$ is an $n\times n$ quantum correlation matrix if and only if there exist unit vectors $u_1,\ldots,u_n,v_1,\ldots,v_n$ in some real Hilbert space $\mathrm{E}$ such that, for each $1\leq i,j\leq n$, $\tau_{i,j}=\langle u_i,v_j\rangle$.

As a particular case, we will say that such $n\times n$ real matrix $\tau$ is a \textit{classical correlation matrix} if the measurement procedure described above can be explained by means of a local hidden variable model \cite{Bell}. Alternatively, this means that $\tau$ is a convex combination of $n\times n$ sign matrices $\varsigma$ defined by: for each $1\leq i,j\leq n$, $\varsigma_{i,j}=\alpha_i\beta_j$ with $\alpha_i,\beta_j\in\{\pm 1\}$.

\subsection{Summary of our main results}

In the next subsection we will give proper definitions of the tensor norm formalism which is the natural framework for studying such correlation matrices. Here we just briefly present the few notation we need in order to state our main results.

Given a correlation matrix $\tau$, considered as an element of $\ell_\infty^n\otimes\ell_\infty^n$, there are norms which measure its {\em quantumness} and its {\em classicality}. These are known as the gamma-$2$ norm $\gamma_2(\cdot)$ and the projective norm $\|\cdot\|_{\ell_\infty^n\otimes_\pi \ell_\infty^n}$, respectively. Specifically, $\tau$ is a quantum correlation matrix if and only if $\gamma_2(\tau)\leq 1$ and $\tau$ is a classical correlation matrix if and only if $\|\tau\|_{\ell_\infty^n\otimes_\pi \ell_\infty^n}\leq 1$.

The main result of our paper is the one below. Before stating it, let us clarify once and for all some conventions that we use repeatedly in the remainder of this paper. We are usually concerned with the asymptotic behaviour of the above-mentioned (and other) norms, i.e. when the size $n$ of the considered matrix grows. In that setting, we use the following (standard) notation: if $f(n),g(n)$ are quantities depending on $n$, $f(n)=o(g(n))$, resp. $f(n)\sim g(n)$, means that the ratio $f(n)/g(n)$ goes to $0$, resp. $1$, as $n$ goes to infinity. Also, when we say that, as $n$ grows, some random event holds ``with high probability'', it means ``with probability at least $1-o(1)$'' (more often than not, the $o(1)$ can in fact be shown to decay even exponentially with $n$).

\begin{theorem} \label{th:q-c-norms:intro}
Let $T$ be an $n\times n$ random matrix with bi-orthogonally invariant distribution, and assume that there exists a constant $r>0$ such that, with high probability as $n\rightarrow\infty$, $\|T\|_{\infty}\leq \left(r +o(1)\right)\|T\|_1/n$. Then, with high probability as $n\rightarrow\infty$,
\begin{equation} \label{eq:q-c-norms} \|T\|_{\ell_{\infty}^n\otimes_{\pi}\ell_{\infty}^n} \geq \left(\sqrt{\frac{16}{15}} - o(1)\right)\gamma_2(T), \end{equation}
so that $\tau=T/\gamma_2(T)$ is an $n\times n$ correlation matrix which is quantum by construction and non-classical with high probability as $n\rightarrow\infty$.
\end{theorem}

Informally stated, this means the following: if $T$ is a random matrix with enough invariance and a flat enough spectrum, then with high probability, its quantum norm is strictly smaller, with constant separation, than its classical norm. As already mentioned, many ``usual'' random matrix models satisfy these two technical assumptions. Moreover, in the proof of the result we see that we can explicitly show a Bell inequality witnessing this separation: if $T=U\Sigma V^t$ is the singular value decomposition of $T$, then $UV^t$ is a Bell inequality which will work with high probability.

The problem of describing the set of quantum correlations was already considered by Tsirelson \cite{Tsi}, who showed its complex geometry. More importantly, the study of those quantum correlations which are not classical can be understood as the study of the location of the classical correlation polytope inside the quantum correlation set. A natural interpretation of Theorem \ref{th:q-c-norms:intro} gives some information along these lines: in almost all directions considered at random in $\R^{n^2}$, the border of the classical polytope is strictly inside the quantum set (see Corollary \ref{cor:gaussian} for a more precise statement).

To prove Theorem \ref{th:q-c-norms:intro} we control both the quantum and the classical norms of $T$. On the one hand, we calculate the quantum norm of $T$: it is with high probability $\|T\|_1/n$. We believe this is a relevant result in its own right, since, to the best of our knowledge, until the moment there were no known results about the quantum norm of a random matrix. On the other hand, we prove a non trivial lower bound on the classical norm of $T$: it is with high probability greater than $\sqrt{16/15}\,\|T\|_1/n$. The main technical tool in order to do so is a fine upper bound on the norm dual to the projective norm (the so-called injective norm) of a Haar distributed orthogonal matrix. The result appears in Proposition \ref{prop:l_1-l_1}, and might be of independent interest, even though we do not have a proper proof of its tightness so far (despite strong numerical evidence).

Finally, using the techniques and results developed in this paper, we obtain a better understanding and an improvement of the results in \cite{GJPV}. In that paper, the following model of random correlation matrices was considered: Let $G,H$ be two independent $n\times m$ Gaussian matrices, and $\widetilde{G},\widetilde{H}$ be their normalized versions (i.e. with rows being unit vectors in $\R^m$). Note that this way of sampling matrix correlations is very natural if one follows the  characterization of quantum correlations due to Tsirelson mentioned above. Then the $n\times n$ correlation matrix $\tau= \widetilde{G} \widetilde{H}^t$ is a quantum correlation matrix by construction, and with probability tending to $1$ as $n$ grows to infinity, $\tau$ is (i)  classical if $m/n\geq 2$ and (ii) not classical if $m/n\leq 0.004$. Our newly developed methods allow us to calculate the quantum norm of $\tau$ and to improve the bound (ii) to $m/n\leq 0.1269$, significantly closing the existing gap.

\subsection{Correlation matrices and tensor norms}

Deciding whether a given matrix $\tau$ corresponds to a correlation matrix (either classical or quantum) can be mathematically neatly written. In order to do so, we need to recall first the definition of several tensor norms. Indeed, a real matrix of size $n\times n$ can always be viewed as an element of $\R^n\otimes\R^n$, by identifying, for each $1\leq i,j\leq n$, $\ketbra{e_i}{e_j}$ with $e_i\otimes e_j$ (where $\{e_1,\ldots,e_n\}$ denotes the canonical orthonormal basis of $\R^n$). This allows to define, for any $n\times n$ real matrix $A$, its injective tensor norm on $\ell_1^n\otimes\ell_1^n$ as
\begin{equation} \label{def:l_1-l_1} \|A\|_{\ell_1^n\otimes_{\epsilon}\ell_1^n} := \sup \left\{\, \sum_{i,j=1}^n A_{i,j}\alpha_i\beta_j \st \alpha_i,\beta_j\in\{\pm 1\} \,\right\}. \end{equation}

Denoting the set of norm one vectors of the Hilbert space $E$ by $S_E$, we can also define
the so-called $\gamma_2^*$ norm of $A$ as
\begin{equation} \label{def:gamma_2*} \gamma_2^*(A) := \sup \left\{\, \sum_{i,j=1}^n A_{i,j}\langle u_i,v_j\rangle \st u_i,v_j\in S_{\mathrm{E}} \,\right\}. \end{equation}
By the preceding explanations, we see that $\tau$ belongs to the set $\mathcal{C}$, resp. $\mathcal{Q}$, of classical, resp. quantum, correlation matrices if and only if, for all $A$ satisfying $\|A\|_{\ell_1^n\otimes_{\epsilon}\ell_1^n}\leq 1$, resp. $\gamma_2^*(A)\leq 1$, we have $\langle \tau,A \rangle :=\Tr(\tau A^t)\leq 1$ (where $A^t$ stands for the transposition of $A$, in the canonical basis). Rephrasing, this means that $\tau\in\mathcal{C}$, resp. $\tau\in\mathcal{Q}$, if and only if $\tau$ is in the unit ball of the norm dual to the $\ell_1^n\otimes_{\epsilon}\ell_1^n$ norm, resp. the $\gamma_2^*$ norm (where duality is meant with respect to the scalar product $\langle \cdot,\cdot\rangle$ introduced above).

The norm dual to $\|\cdot\|_{\ell_1^n\otimes_{\epsilon}\ell_1^n}$ is the projective tensor norm on $\ell_{\infty}^n\otimes\ell_{\infty}^n$, which can be written as
\begin{equation} \label{def:l_infty-l_infty} \|\tau\|_{\ell_{\infty}^n\otimes_{\pi}\ell_{\infty}^n} := \inf \left\{\, \sum_{k=1}^N \|x_k\|_{\infty}\|y_k\|_{\infty} \st \tau=\sum_{k=1}^N x_k\otimes y_k\,\right\}. \end{equation}
While the norm dual to $\gamma_2^*(\cdot)$ is, as the notation suggests, the so-called $\gamma_2$ norm, which can be written as
\begin{equation} \label{def:gamma_2} \gamma_2(\tau) := \inf \left\{\, \|X\|_{\ell_2{\rightarrow}\ell_{\infty}^n} \|Y\|_{\ell_1^n{\rightarrow}\ell_2} \st \tau=XY \,\right\}, \end{equation}
where, denoting by $R_i(X)$ the rows of an $n\times m$ matrix $X$ and by $C_j(Y)$ the columns of an $m\times n$ matrix $Y$, we have
\[ \|X\|_{\ell_2{\rightarrow}\ell_{\infty}^n}=\underset{1\leq i\leq n}{\max}\|R_i(X)\|_2\ \text{and}\ \|Y\|_{\ell_1^n{\rightarrow}\ell_2}=\underset{1\leq j\leq n}{\max}\|C_j(Y)\|_2. \]

Hence recapitulating, the reason why we are interested in these norms in the context of correlation matrices is clear, since given any $n\times n$ real matrix $\tau$,
\[ \tau\in\mathcal{C}\ \Leftrightarrow\ \|\tau\|_{\ell_{\infty}^n\otimes_{\pi}\ell_{\infty}^n}\leq 1\ \text{   and    }\ \tau\in\mathcal{Q}\ \Leftrightarrow\ \gamma_2(\tau)\leq 1. \]

It is a well known fact that these two norms cannot differ {\em too much}. Precisely, Grothendieck's inequality (see e.g. \cite{Pis}, Section 3, for further comments and proofs) tells us, for any $n\times n$ real matrix $T$,
\begin{equation} \label{eq:groth} \gamma_2(T)\leq \|T\|_{\ell_{\infty}^n\otimes_{\pi}\ell_{\infty}^n}\leq K_G\,\gamma_2(T), \end{equation}
where $K_G$ is the (real) Grothendieck constant, whose exact value is unknown (but somewhere between $1.67696...$ and $1.78221...$).

Hence, what Theorem \ref{th:q-c-norms:intro} actually tells us is that the first inequality in equation \eqref{eq:groth} can be improved to $\sqrt{16/15}\,\gamma_2(T)\leq \|T\|_{\ell_{\infty}^n\otimes_{\pi}\ell_{\infty}^n}$ when one is interested in an inequality which is not necessarily true for {\em any} matrix but only for {\em typical} ones (in a sense to be made precise).\

\subsection{Two needed technical lemmas}

We will crucially exploit at several occasions later on the following fact: for a bi-orthogonally invariant random matrix $T$, the three random matrices $U,V,\Sigma$ associated to the singular value decomposition $T=U\Sigma V^t$ are distributed independently of each other and $U, V$ are Haar distributed orthogonal matrices. The precise statement we will use is the following. A proof for it can be seen in \cite{GJPV}, Proposition 1.5.

\begin{lemma}\label{SVD}
Let $T$ be an $n \times n$ random matrix with bi-orthogonally invariant distribution, in some probability space $(\Omega,\mathbb{P})$. Then there exist three probability spaces $(\Omega_1, P_1)$, $(\Omega_2, P_2)$, $(\Omega_3, P_3)$ and three $n\times n$ random matrices $U(\omega_1)$, $\Sigma(\omega_2)$, $V(\omega_3)$ in $(\Omega_1, P_1)$,  $(\Omega_2, P_2)$, $(\Omega_3, P_3)$ respectively, such that
\begin{enumerate}
\item[(i)] The matrices $U,V$ are Haar distributed orthogonal matrices on $\R^n$.
\item[(ii)] The matrix $\Sigma$ is diagonal and all of its elements are non-negative.
\item[(iii)] The random matrix $T'(\omega_1, \omega_2, \omega_3)=U(\omega_1)\Sigma(\omega_2)V^t(\omega_3)$ follows the same distribution as the random matrix $T(\omega)$.
\end{enumerate}
\end{lemma}

We also recall here Levy's Lemma \cite{Levy}, which guarantees that regular enough (i.e. Lipschitz) functions on the unit sphere typically concentrate around their average behaviour. Note that when we talk about average ($\mathbb E f$) or median ($M_f$) of a function $f$  and probability of deviating from it, these are always computed with respect to the uniform probability measure on the unit sphere. We will actually use two versions of Levy's Lemma: in Section \ref{sec:quantum} a ``rough'' one will be enough for our purposes, while in Section \ref{sec:classical} we will really need a ``tight'' one . We consequently state both here, the reader being for instance referred to \cite{ASbook}, Chapter 5, for detailed comments and proofs. We just point out that the second inequality in Lemma \ref{levi} below follows from the fact that the volume (again with respect to the uniform probability measure) of a spherical cap in $S^{n-1}$ with geodesic radius $\theta$ is upper bounded by $(\sin\theta)^{n-1}/2$ (see e.g. \cite{ASbook}, Proposition 5.1, for a proof).

\begin{lemma}\label{levi}
Let $f:S^{n-1}\longrightarrow \mathbb R$ be an $L$-Lipschitz function. Then,
\[ \forall\ \epsilon>0,\ \mathbb P\left(f>\mathbb E f + \epsilon L\right)\leq e^{-cn\epsilon^2}, \]
where $c>0$ is a universal constant. Furthermore,
\[ \forall\ 0<\theta<\frac{\pi}{2},\ \mathbb P\left(f>M_f+ (\cos\theta) L\right) \leq \frac{1}{2}(\sin\theta)^{n-1} . \]
\end{lemma}

\

\section{Quantum norm of a random bi-orthogonally invariant matrix}\label{sec:quantum}

In this section we describe how to calculate, with high probability, the quantum norm of a random bi-orthogonally invariant matrix in terms of its trace norm.

\subsection{$\gamma_2^*$ norm of an orthogonal matrix}

We begin with a simple result about the $\gamma_2^*$ norm of an orthogonal matrix. Note that, opposite to most of the paper, this is not a probabilistic statement, but one which holds for every orthogonal matrix (not just a Haar distributed one).

\begin{lemma}\label{gammastar} Let $O$ be an orthogonal matrix on $\R^n$. Then, $\gamma_2^*(O)=n$. \end{lemma}

\begin{proof} Norm duality between $\gamma_2^*(\cdot)$ and $\gamma_2(\cdot)$ tells us that
\[ \gamma_2^*(O)=\sup \{ \langle O, \tau\rangle \st \gamma_2(\tau)\leq 1\}. \]

To prove that $\gamma_2^*(O)\leq n$, let us consider one such $\tau$. Since $\gamma_2(\tau)\leq 1$, we know that there exist vectors $u_1, \ldots, u_n, v_1, \ldots, v_n$ in the unit ball of some $\R^m$ such that, for every $1\leq i, j\leq n$, $\tau_{i,j}=\langle u_i,v_j\rangle$. That is, if we define the $n\times m$ matrix $U$, resp. $V$, whose rows are the vectors $u_1, \ldots, u_n$, resp. $v_1, \ldots, v_n$, we have $\tau=UV^t$. Then,
\begin{equation}\label{omega2}
\langle O,\tau\rangle=\Tr (OVU^t) =\sum_{i,j=1}^n (OV)_{i,j} U_{j,i}\leq \left(\sum_{i,j=1}^n (OV)_{i,j}^2\right)^{1/2}  \left(\sum_{i,j=1}^n U_{j,i}^2\right)^{1/2}=n,
\end{equation}
where the next to last inequality follows from Cauchy-Schwarz inequality and the last inequality uses the fact that $O$ is an orthogonal matrix, hence an isometry.

To show the other direction, namely $\gamma_2^*(O)\geq n$, note that equality in equation \eqref{omega2} is attained for $\tau=O$, which clearly verifies $\gamma_2(O)\leq 1$.
\end{proof}

\subsection{$\gamma_2$ norm of a random bi-orthogonally invariant matrix}

The following well-known relation between the trace norm of a matrix and its $\gamma_2$ norm follows immediately from Lemma \ref{gammastar}. Note that, again, this result is true for any matrix, not just with high probability for certain random matrices.

\begin{proposition} \label{prop:lower-q}
For every matrix $T$ on $\R^n$,
\[ \gamma_2(T)\geq \frac{\|T\|_1}{n}. \]
\end{proposition}

\begin{proof}
Let $T=U\Sigma V^t$ be the singular value decomposition of $T$. Applying Lemma \ref{gammastar} and norm duality between $\gamma_2(\cdot)$ and $\gamma_2^*(\cdot)$ we get
\[ \gamma_2(T)  \geq   \frac{\langle T, UV^t\rangle}{\gamma_2^*(UV^t)}=\frac{\Tr(\Sigma)}{n} =\frac{\|T\|_1}{n}, \]
as wanted.
\end{proof}

Our main result in this section says that for random bi-orthogonally invariant matrices, the above inequality is, with high probability as $n\rightarrow\infty$, essentially tight.

\begin{proposition} \label{prop:upper-q}
Let $T$ be an $n\times n$ random matrix with bi-orthogonally invariant distribution, and assume that there exists a constant $r>0$ such that, with high probability as $n\rightarrow\infty$, $\|T\|_{\infty}\leq \left(r +o(1)\right)\|T\|_1/n$.
Then, with high probability as $n\rightarrow\infty$,
\[ \gamma_2(T) \leq \left(1 + o(1)\right)\frac{\|T\|_1}{n}. \]
\end{proposition}

\begin{proof}
Using Lemma \ref{SVD}, we can consider $T$ to be of the form $T=U\Sigma V^t$, where  $U,V,\Sigma$ are independent random matrices and  $U,V$ are Haar distributed orthogonal matrices on $\R^n$.

Now, let $S$ be a fixed diagonal matrix on $\R^n$ with positive eigenvalues $S_1,\ldots,S_n$, and set $r=n\|S\|_{\infty}/\|S\|_1$. Consider next the random matrices $X=U\sqrt{S}$ and $Y=\sqrt{S}V^t$, for $U,V$ independent Haar distributed orthogonal matrices on $\R^n$. Observe that, for each $1\leq i\leq n$, the $i^{\text{th}}$ row of $X$, $R_i(X)=(\sqrt{S_1}U_{i,1},\ldots,\sqrt{S_n}U_{i,n})$, is distributed as $(\sqrt{S_1}\psi_1,\ldots,\sqrt{S_n}\psi_n)$ for $\psi\in S^{n-1}$ a uniformly distributed unit vector.

Moreover, the function $\psi\in S^{n-1}\mapsto \sum_{k=1}^nS_k\psi_k^2\in\R$ has average (with respect to the uniform probability measure) equivalent to $\|S\|_1/n$, as $n\rightarrow\infty$, and Lipschitz constant upper bounded by $2\|S\|_{\infty}$. The first claim can easily be seen to hold if viewing $\psi$ as $g/\|g\|_2$, where $g\in\R^n$ has independent mean $0$ and variance $1$ Gaussian entries. Indeed, such $g$ satisfies the two properties that $\|g\|_2$ is independent from $g/\|g\|_2$ and equivalent to $\sqrt{n}$ as $n\rightarrow\infty$. The second claim follows from the chain of (in)equalities
\begin{align*} 
\left|\sum_{k=1}^nS_k\psi_k^2-\sum_{k=1}^nS_k\varphi_k^2\right| & = \left|\sum_{k=1}^nS_k(\psi_k+\varphi_k)(\psi_k-\varphi_k)\right|\\
& \leq \|S\|_{\infty}\|\psi+\varphi\|_2\|\psi-\varphi\|_2\\ 
& \leq 2\|S\|_{\infty}\|\psi-\varphi\|_2, 
\end{align*}
where the inequality before last follows from H\"{o}lder's inequality and the last inequality follows from the triangle inequality. Hence by Levy's lemma, recalled as Lemma \ref{levi}, we have that there exists a universal constant $c>0$ such that
\[ \forall\ \epsilon>0,\ \P\left(\sum_{k=1}^nS_k\psi_k^2>(1+\epsilon)\frac{\|S\|_1}{n}\right) \leq e^{-cn\epsilon^2/r^2}. \]

Coming back to our initial problem, what we have equivalently shown is that, there exists a universal constant $c>0$ such that, for each $1\leq i\leq n$,
\[ \forall\ \epsilon>0,\ \P\left(\|R_i(X)\|_2^2>(1+\epsilon)\frac{\|S\|_1}{n}\right) \leq e^{-cn\epsilon^2/r^2}, \]
so that by the union bound
\[ \forall\ \epsilon>0,\ \P\left(\exists\ 1\leq i\leq n:\ \|R_i(X)\|_2^2>(1+\epsilon)\frac{\|S\|_1}{n}\right) \leq ne^{-cn\epsilon^2/r^2}. \]
Obviously the same reasoning holds for the columns of $Y$, $C_j(Y)=(\sqrt{S_1}V_{1,j},\ldots,\sqrt{S}_nV_{n,j})$, $1\leq j\leq n$. Hence what we eventually obtain is
\[ \forall\ \epsilon>0,\ \P\left(\exists\ 1\leq i,j\leq n:\ \|R_i(X)\|_2\|C_j(Y)\|_2>(1+\epsilon)\frac{\|S\|_1}{n}\right) \leq 2ne^{-cn\epsilon^2/r^2}. \]
Indeed, defining the three events
\begin{align*} & \mathcal{A}_1 =\text{``}\,\exists\ 1\leq i\leq n:\ \|R_i(X)\|_2^2>(1+\epsilon)\|S\|_1/n\,\text{''},\\
& \mathcal{A}_2 =\text{``}\,\exists\ 1\leq j\leq n:\ \|C_j(Y)\|_2^2>(1+\epsilon)\|S\|_1/n\,\text{''},\\
& \mathcal{B} =\text{``}\,\exists\ 1\leq i,j\leq n:\ \|R_i(X)\|_2\|C_j(Y)\|_2>(1+\epsilon)\|S\|_1/n\,\text{''},
\end{align*}
we clearly have $\mathcal{B}\subset \mathcal{A}_1\cup\mathcal{A}_2$, so that
\[ \P(\mathcal{B})\leq \P(\mathcal{A}_1\cup\mathcal{A}_2)\leq \P(\mathcal{A}_1)+\P(\mathcal{A}_2). \]

We finally just have to recall that $\Sigma$ is independent from $U,V$, and satisfying by assumption that there exists $r>0$ such that $n\|\Sigma\|_{\infty}/\|\Sigma\|_1\leq r+o(1)$ with high probability as $n\rightarrow\infty$. By what precedes, we therefore see that, with high probability as $n\rightarrow\infty$,
\[ \forall\ 1\leq i,j\leq n,\ \left\|R_i\left(U\sqrt{\Sigma}\right)\right\|_2\left\|C_j\left(\sqrt{\Sigma}V^t\right)\right\|_2 \leq \left(1+o(1)\right)\frac{\|\Sigma\|_1}{n}. \]
By definition of $\gamma_2(\cdot)$ (see equation \eqref{def:gamma_2}), this implies exactly the announced result.
\end{proof}

Putting together Propositions \ref{prop:lower-q} and \ref{prop:upper-q}, we immediately obtain the conclusion below on the $\gamma_2$ norm of a random bi-orthogonally invariant matrix.

\begin{theorem} \label{th:quantum-norm}
Let $T$ be an $n\times n$ random matrix with bi-orthogonally invariant distribution, and assume that there exists a constant $r>0$ such that, with high probability as $n\rightarrow\infty$, $\|T\|_{\infty}\leq \left(r +o(1)\right)\|T\|_1/n$.
Then, with high probability as $n\rightarrow\infty$,
\[ \gamma_2(T) = \left(1 \pm o(1)\right)\frac{\|T\|_1}{n}. \]
\end{theorem}

\

\section{Classical norm of a random bi-orthogonally invariant matrix} \label{sec:classical}

In this section we describe how to lower bound, with high probability, the classical norm of a random bi-orthogonally invariant matrix in terms of its trace norm.

\subsection{$\ell_1^n\otimes_{\epsilon}\ell_1^n$ norm of a random orthogonal matrix} In order to achieve our goal, we will need first an upper bound on the $\ell_1^n\otimes_{\epsilon}\ell_1^n$ norm of a Haar distributed orthogonal matrix, which might be of independent interest.

\begin{proposition} \label{prop:l_1-l_1}
Let $O$ be a Haar distributed orthogonal matrix on $\R^n$. Then, with high probability as $n\rightarrow\infty$,
\begin{equation} \label{eq:l_1-l_1} \|O\|_{\ell_1^n\otimes_{\epsilon}\ell_1^n} \leq \left(\sqrt{\frac{15}{16}}+o(1)\right)n. \end{equation}
\end{proposition}

\begin{proof}
Define the random vector $x\in\R^n$ by, for each $1\leq i\leq n$, $x_i=\sum_{j=1}^nO_{i,j}$. Then, observe that $x=\sqrt{n}\psi$ with $\psi\in S^{n-1}$ a uniformly distributed unit vector. Indeed, since $O$ is orthogonal, we have firstly that $\|x\|=\sqrt{n}$, and since $O$ is orthogonal and Haar distributed, we have secondly that for any orthogonal matrix $V$ on $\R^n$, $x$ and $Vx$ have the same distribution. Therefore, for any $0<\theta<\pi/2$,
\[ \P\left(\sum_{i,j=1}^nO_{i,j}>(\cos\theta)n\right) = \P\left(\sum_{i=1}^n x_i>(\cos\theta)n\right) = \P\left(\sum_{i=1}^n \psi_i>(\cos\theta)\sqrt{n}\right). \]
Now, the function $\psi\in S^{n-1} \mapsto \sum_{i=1}^n \psi_i\in\R$ has median (with respect to the uniform probability measure) equal to $0$, and Lipschitz constant upper bounded by $\sqrt{n}$ (the latter claim follows from Cauchy-Schwarz inequality). Hence by Levy's lemma, recalled as Lemma \ref{levi}, we obtain that
\[ \P\left(\sum_{i=1}^n \psi_i>(\cos\theta)\sqrt{n}\right) \leq \frac{1}{2}(\sin\theta)^{n-1}. \]
Recapitulating, what we have shown so far is that, for any fixed $\alpha,\beta\in\{\pm 1\}^n$,
\[ \P\left(\sum_{i,j=1}^n\alpha_i\beta_jO_{i,j}>(\cos\theta)n\right) \leq \frac{1}{2}(\sin\theta)^{n-1}. \]
Consequently, by the union bound, we have
\[ \P\left(\exists\ \alpha,\beta\in\{\pm 1\}^n:\ \sum_{i,j=1}^n\alpha_i\beta_jO_{i,j}>(\cos\theta)n\right) \leq 4^n\frac{1}{2}(\sin\theta)^{n-1}. \]
By definition of $\|\cdot\|_{\ell_1^n\otimes_{\epsilon}\ell_1^n}$ (see equation \eqref{def:l_1-l_1}), this means precisely that, for any $0<\theta<\pi/2$,
\begin{equation} \label{eq:l_1-l_1'} \P\left( \|O\|_{\ell_1^n\otimes_{\epsilon}\ell_1^n} \leq (\cos\theta)n \right) \geq 1-2(4\sin\theta)^{n-1}. \end{equation}
In order to conclude, we thus just have to observe that, as soon as $\theta<\arcsin(1/4)$, i.e. equivalently $\cos\theta>\sqrt{15/16}$, the right hand-side in equation \eqref{eq:l_1-l_1'} goes to $1$ exponentially fast as $n$ grows.
\end{proof}

Let us comment a bit on Proposition \ref{prop:l_1-l_1}. The first thing that may be worth pointing out is that the $\ell_1^n\otimes_{\epsilon}\ell_1^n$ norm of a matrix is nothing else than its $\ell_{\infty}^n{\rightarrow}\ell_1^n$ norm. We recall that, for any $1\leq p,q\leq \infty$, the $\ell_p^n{\rightarrow}\ell_q^n$ norm is naturally defined as follows: for any matrix $A$ on $\R^n$,
\[ \|A\|_{\ell_p^n{\rightarrow}\ell_q^n}:= \sup\left\{\,\|Ax\|_q \st x\in\R^n,\ \|x\|_p\leq 1\,\right\}. \]
Hence in particular, the $\ell_2^n{\rightarrow}\ell_2^n$ norm is simply the operator norm. Now, by duality between $\|\cdot\|_1$ and $\|\cdot\|_{\infty}$, combined with extremality of sign vectors in the unit ball for $\|\cdot\|_{\infty}$, it is clear that, as claimed,
\begin{align*} 
\|A\|_{\ell_{\infty}^n{\rightarrow}\ell_1^n} & = \sup\left\{\,\|Ax\|_1 \st x\in\R^n,\ \|x\|_{\infty}\leq 1\,\right\} \\
& = \sup \left\{\,\bra{y}A\ket{x} \st x,y\in\{\pm 1\}^n\,\right\} \\
& = \|A\|_{\ell_1^n\otimes_{\epsilon}\ell_1^n}. 
\end{align*}

Then, we may first remark that, by Cauchy-Schwarz inequality, we obviously have: for any matrix $A$ on $\R^n$,
\[ \|A\|_{\ell_{\infty}^n{\rightarrow}\ell_1^n}\leq n\|A\|_{\ell_2^n{\rightarrow}\ell_2^n}. \]
In the case of an orthogonal matrix $O$ on $\R^n$, $\|O\|_{\ell_2^n{\rightarrow}\ell_2^n}=1$, and it therefore always holds that $\|O\|_{\ell_{\infty}^n{\rightarrow}\ell_1^n}\leq n$. Besides, there of course exist orthogonal matrices on $\R^n$ (such as e.g. the identity matrix) for which this inequality is in fact an equality. However, what Proposition \ref{prop:l_1-l_1} tells us is that, for $O$ a Haar distributed orthogonal matrix on $\R^n$, a slightly better upper bound on $\|O\|_{\ell_{\infty}^n{\rightarrow}\ell_1^n}$ actually holds with high probability, namely $\sqrt{15/16}\,n$.

What is more, for $O$ a Haar distributed orthogonal matrix on $\R^n$,
\[ \E\sup\left\{\, \|Ox\|_1 \st x\in\R^n,\ \|x\|_{\infty}\leq 1\,\right\} \geq \sup \left\{\, \E\|Ox\|_1 \st x\in\R^n,\ \|x\|_{\infty}\leq 1\,\right\} \underset{n\rightarrow\infty}{\sim} \sqrt{\frac{2}{\pi}} n. \]
The argument leading to the last equivalence is exactly the same as in the proof of Proposition \ref{prop:l_1-l_1}: for each $x\in\{\pm 1\}^n$, $Ox=\sqrt{n}\psi$ with $\psi\in S^{n-1}$ a uniformly distributed unit vector, and it is well-known that $\E\|\psi\|_1\sim\sqrt{2/\pi}\sqrt{n}$ for such $\psi$. Then, by concentration for Lipschitz functions on the orthogonal group (see e.g. the Appendix in \cite{MM}), this is in fact true not only on average but also with high probability. Together with Proposition \ref{prop:l_1-l_1}, this means that, for $O$ a Haar distributed orthogonal matrix on $\R^n$, with high probability as $n\rightarrow\infty$,
\[ \left(\sqrt{\frac{2}{\pi}}-o(1)\right) n \leq \|O\|_{\ell_1^n\otimes_{\epsilon}\ell_1^n} \leq \left(\sqrt{\frac{15}{16}}+o(1)\right) n. \]
Numerics suggest that the true asymptotic value of $\|O\|_{\ell_1^n\otimes_{\epsilon}\ell_1^n}/n$ would actually be $\sqrt{15/16}$, but we were not able to prove that mathematically.

\subsection{Lower bound on the $\ell_{\infty}^n\otimes_{\pi}\ell_{\infty}^n$ norm of a random bi-orthogonally invariant matrix}

\begin{theorem} \label{th:classical-norm}
Let $T$ be an $n\times n$ random matrix with bi-orthogonally invariant distribution. Then, with high probability as $n\rightarrow\infty$,
\begin{equation} \label{eq:classical-norm} \|T\|_{\ell_{\infty}^n\otimes_{\pi}\ell_{\infty}^n} \geq \left(\sqrt{\frac{16}{15}} - o(1)\right)\frac{\|T\|_1}{n}. \end{equation}
\end{theorem}

\begin{proof}
Let $T=U\Sigma V^t$ be the singular value decomposition of $T$. By norm duality between $\|\cdot\|_{\ell_{\infty}^n\otimes_{\pi}\ell_{\infty}^n}$ and $\|\cdot\|_{\ell_1^n\otimes_{\epsilon}\ell_1^n}$, it is clear that
\begin{equation} \label{eq:bell-ineq} \|T\|_{\ell_{\infty}^n\otimes_{\pi}\ell_{\infty}^n} \geq \frac{\langle T,UV^t\rangle}{\|UV^t\|_{\ell_1^n\otimes_{\epsilon}\ell_1^n}} = \frac{\Tr(\Sigma)}{\|UV^t\|_{\ell_1^n\otimes_{\epsilon}\ell_1^n}}. \end{equation}
Now, by the bi-orthogonally invariance hypothesis on the distribution of $T$, we know from Lemma \ref{SVD} that we can consider $UV^t$ to be,  first, independent from $\Sigma$, and, second, a Haar distributed orthogonal matrix on $\R^n$. The latter fact implies by Proposition \ref{prop:l_1-l_1} that $\|UV^t\|_{\ell_1^n\otimes_{\epsilon}\ell_1^n}\leq \big(\sqrt{15/16}+o(1)\big)n$ with high probability. Inserting this upper bound in equation \eqref{eq:bell-ineq}, and using that $\Tr(\Sigma)=\|T\|_1$ is independent from $\|UV^t\|_{\ell_1^n\otimes_{\epsilon}\ell_1^n}$, yields exactly the announced lower bound.
\end{proof}

An important question at this point is that of optimality in Theorem \ref{th:classical-norm}. There are two places where we might lose something. First of all, we may be doing things roughly when estimating $\|O\|_{\ell_1^n\otimes_{\epsilon}\ell_1^n}$ for $O$ a Haar distributed orthogonal matrix on $\R^n$. Indeed, as we already discussed before, we do not know whether or not the upper bound $\sqrt{15/16}\,n$ provided by Proposition \ref{prop:l_1-l_1} is tight. Second, the choice of the orthogonal matrices appearing in the singular value decomposition of $T$ as Bell functional may not be the best one. We showed it is the optimal choice in order to compute the quantum norm of $T$, but there is a priori no reason it remains so in order to compute its classical norm. 

\

\section{Consequences}
\subsection{Separation between the quantum norm and the classical norm of a random bi-orthogonally invariant matrix}

As a straightforward consequence of Theorems \ref{th:quantum-norm} and \ref{th:classical-norm}, providing estimates on, respectively, the quantum norm and the classical norm of random bi-orthogonally invariant matrices, we obtain Theorem \ref{th:q-c-norms} below.

\begin{theorem} \label{th:q-c-norms}
Let $T$ be an $n\times n$ random matrix with bi-orthogonally invariant distribution, and assume that there exists a constant $r>0$ such that, with high probability as $n\rightarrow\infty$, $\|T\|_{\infty}\leq \left(r +o(1)\right)\|T\|_1/n$. Then, with high probability as $n\rightarrow\infty$,
\begin{equation} \label{eq:q-c-norms} \|T\|_{\ell_{\infty}^n\otimes_{\pi}\ell_{\infty}^n} \geq \left(\sqrt{\frac{16}{15}} - o(1)\right)\gamma_2(T), \end{equation}
so that $\tau=T/\gamma_2(T)$ is an $n\times n$ correlation matrix which is quantum by construction and non-classical with high probability as $n\rightarrow\infty$.
\end{theorem}

One especially interesting case of Theorem \ref{th:q-c-norms} is when $T=G/\sqrt{n}$, for $G$ an $n\times n$ Gaussian matrix, that is, a matrix with independent mean $0$ and variance $1$ real Gaussian entries. Such $T$ has a bi-orthogonally invariant distribution and satisfies $\|T\|_1\sim (8/3\pi)n$ and $\|T\|_{\infty}\sim 2$ with high probability as $n\rightarrow\infty$ (see e.g. \cite{AGZ}, Chapter 2, for a proof), so it indeed fulfills the hypotheses of Theorem \ref{th:q-c-norms} (with $r=3\pi/4$). As a consequence, the matrix $\tau=T/\gamma_2(T)$ is in $\mathcal{Q}$ and with high probability not in $\mathcal{C}$. Now, such $\tau$ is actually by construction uniformly distributed on the border of $\mathcal{Q}$. So what this result tells us is that, in most directions in $\R^{n^2}$, the borders of $\mathcal{Q}$ and $\mathcal{C}$ do not coincide. Additionally, we can exhibit a Bell functional separating these two borders most of the time, namely the matrix $UV^t$ where $\tau=U\Sigma V^t$ is the singular value decomposition of the considered direction. Note that this orthogonal matrix $UV^t$ is unique if $\tau$ is invertible. Hence, since Gaussian matrices are invertible with probability one, given a realization of $\tau$ the Bell functional is fixed and explicit. Let us summarize this discussion in Corollary \ref{cor:gaussian} below.

\begin{corollary} \label{cor:gaussian}
Let $\tau$ be uniformly distributed on the border of the set of $n\times n$ quantum correlation matrices. Then as $n\rightarrow\infty$, $\tau$ is with high probability outside the set of $n\times n$ classical correlation matrices. Furthermore, if $\tau=U\Sigma V^t$ is the singular value decomposition of $\tau$, then with high probability, $A=UV$ is a Bell functional certifying it, since
\[ \max\left\{\,\Tr(\tau'A^t) \st \tau'\in\mathcal{C}\,\right\} \leq \left(\sqrt{\frac{15}{16}}+o(1)\right)\Tr(\tau A^t) < \Tr(\tau A^t). \]
\end{corollary}

A (weaker) consequence of Corollary \ref{cor:gaussian} is in terms of average widths of the convex sets $\mathcal{C}$ and $\mathcal{Q}$. In fact, it is rather a statement about the average widths of the convex sets $\mathcal{C}^*$ and $\mathcal{Q}^*$, dual to $\mathcal{C}$ and $\mathcal{Q}$, that we will be able to derive here. Before establishing it, let us recall the concept of \textit{mean width} from classical convex geometry. Given a convex body $K\subset\R^k$, denoting by $\psi$ a uniformly distributed unit vector in $\R^k$, its mean width is defined as
\[ w(K) := \E \sup \left\{\, \langle\psi,x\rangle \st x\in K \,\right\}. \]
This spherical averaging can be replaced by a (usually more convenient to deal with) Gaussian averaging, just noticing that, denoting by $g$ a Gaussian vector in $\R^k$,
\[ w(K) \underset{k\rightarrow\infty}{\sim} \frac{1}{\sqrt{k}}\, \E \sup \left\{\, \langle g,x\rangle \st x\in K \,\right\}. \]
Correlation matrices of size $n\times n$ can be seen as convex bodies in $\R^{n^2}$. We thus have, letting $\mathcal{T}$ be either $\mathcal{Q}$ or $\mathcal{C}$, $\mathcal{T}^*$ be its dual, and denoting by $G$ a Gaussian matrix on $\R^n$,
\[ w(\mathcal{T}^*) \underset{n\rightarrow\infty}{\sim} \frac{1}{n}\, \E \sup \left\{\, \Tr(GA^t) \st A\in\mathcal{T}^* \,\right\}. \]

\begin{corollary} \label{cor:mean-width}
The duals of the sets of $n\times n$ quantum and classical correlation matrices satisfy the following mean width estimates, as $n\rightarrow\infty$,
\[ w(\mathcal{Q}^*) \sim \frac{8}{3\pi}\frac{1}{\sqrt{n}}\ \text{and}\ w(\mathcal{C}^*) \geq \left(\sqrt{\frac{16}{15}}-o(1)\right) \frac{8}{3\pi}\frac{1}{\sqrt{n}}. \]
\end{corollary}

\begin{proof}
As explained in the Introduction, the convex bodies $\mathcal{Q}^*$ and $\mathcal{C}^*$ are simply the unit balls for the norms $\gamma_2^*(\cdot)$ and $\|\cdot\|_{\ell_1^n\otimes_{\epsilon}\ell_1^n}$, respectively. Hence, by definition of the mean width, we have, denoting by $G$ a Gaussian matrix on $\R^n$,
\begin{align*} w(\mathcal{Q}^*) & \underset{n\rightarrow\infty}{\sim} \frac{1}{n}\,\E \sup \left\{\, \Tr(GA^t) \st \gamma_2^*(A)\leq 1\, \right\} = \frac{1}{n}\,\E\gamma_2(G)\\
w(\mathcal{C}^*) & \underset{n\rightarrow\infty}{\sim} \frac{1}{n}\,\E \sup \left\{\, \Tr(GA^t) \st \|A\|_{\ell_1^n\otimes_{\epsilon}\ell_1^n}\leq 1\, \right\} = \frac{1}{n}\,\E\|G\|_{\ell_{\infty}^n\otimes_{\pi}\ell_{\infty}^n}, \end{align*}
where we just used that $\gamma_2(\cdot)$ is dual to $\gamma_2^*(\cdot)$ and $\|\cdot\|_{\ell_{\infty}^n\otimes_{\pi}\ell_{\infty}^n}$ is dual to $\|\cdot\|_{\ell_1^n\otimes_{\epsilon}\ell_1^n}$.

Now, we know by Theorem \ref{th:quantum-norm} that, as $n\rightarrow\infty$, $\E\gamma_2(G)\sim (8/3\pi)\sqrt{n}$ and by Theorem \ref{th:classical-norm} that, as $n\rightarrow\infty$, $\E\|G\|_{\ell_{\infty}^n\otimes_{\pi}\ell_{\infty}^n} \geq \big(\sqrt{16/15}-o(1)\big)(8/3\pi)\sqrt{n}$, which concludes the proof of Corollary \ref{cor:mean-width}.
\end{proof}

In words, Corollary \ref{cor:mean-width} tells us the following: the ratio $w(\mathcal{C}^*)/w(\mathcal{Q}^*)$ is asymptotically at least $\sqrt{16/15}$, hence stays lower bounded away from $1$. This result complements earlier findings on the mean width of $\mathcal{Q}$ and $\mathcal{C}$. Indeed, it was proved in \cite{Amb+} that, as $n\rightarrow\infty$,
\[ w(\mathcal{Q}) \sim 2\sqrt{n}\ \text{and}\ w(\mathcal{C}) \leq \left(2\sqrt{\ln 2}+o(1)\right)\sqrt{n}. \]
The results established in \cite{Amb+} are actually about random Bernoulli Bell functionals instead of random Gaussian ones, but the estimate on their quantum value and the upper bound on their classical value remain true in the Gaussian case. It was therefore already known that the ratio $w(\mathcal{Q})/w(\mathcal{C})$ is asymptotically at least $1/\sqrt{\ln 2}$, hence it stays lower bounded away from $1$.

\subsection{Non-local quantum correlations from random unit vectors}
In this subsection we use the previous estimates to study the random quantum correlation matrices considered in \cite{GJPV}, i.e.
\begin{align}\label{correlation inner}
\tau=(\langle u_i,v_j\rangle)_{i,j=1}^n,
\end{align}
where the vectors $u_1,\dots, u_n, v_1,\dots, v_n$ are independently identically uniformly distributed in the unit sphere of $\mathbb R^m$, for some $m\in\N$. It was proved in \cite{GJPV} that the asymptotic (non-)local behaviour of such (by construction quantum) correlation matrix depends on the limit ratio $\alpha=\lim_{n\rightarrow\infty}m/n$. The results in \cite {GJPV} showed that a correlation matrix sampled in this way is asymptotically with high probability non-local if $\alpha <0.004$, and local if $\alpha>2$. Here we improve the bound on the non-local regime up to a limit ratio of $\alpha<0.1269$.

\begin{theorem}\label{prodgauss}
Let $n,m$ be two natural numbers and set $\alpha=m/n$. Let us consider $2n$ vectors $u_1,\dots, u_n, v_1,\dots, v_n$ sampled independently according to the uniform measure in the unit sphere of $\mathbb R^m$ and let us denote by $\tau=(\langle u_i,v_j\rangle)_{i,j=1}^n$ the corresponding quantum correlation matrix. Then, there exists an $\alpha_0$ such that if $\alpha< \alpha_0\approx0.1269$, $\tau$ is non-local with probability tending to $1$ as $n$ tends to infinity.
\end{theorem}

Before proving Theorem \ref{prodgauss} we want to mention that the results presented in the previous sections not only allow us to improve the bounds  in \cite{GJPV} but they go deeper in the understanding of the correlations of the form (\ref{correlation inner}). Indeed, since such matrix $\tau$ has a bi-orthogonally invariant distribution, the same reasonings as in Theorem \ref{th:quantum-norm} can be done to estimate $\gamma_2(\tau)$ as a function of $\|\tau\|_1$, and there are explicit formulae for the asymptotic behaviour of this last norm. That is, we can compute the asymptotic behaviour of $\gamma_2(\tau)$. Moreover, although we still do not know whether $\tau$ is local or not for every ratio $m/n$, Theorem \ref{th:q-c-norms:intro} guarantees that its $\ell_{\infty}^n\otimes_{\pi}\ell_{\infty}^n$ norm is with high probability strictly larger than its $\gamma_2$ norm for every ratio $m/n$. That is, if we divide  $\tau$ by its $\gamma_2$ norm, we will typically obtain a quantum non-local correlation.

The idea of the proof of Theorem \ref{prodgauss} is the following: We will approximate the quantum correlation $\tau$ by a product of two (renormalized) Gaussian matrices $GH^t/m$, such that
$$ \lim_{m,n\rightarrow\infty} \left\|\tau- \frac{1}{m}GH^t\right\|_{\ell_{\infty}^n\otimes_{\pi}\ell_{\infty}^n} =0. $$
Hence, $\tau$ will be, asymptotically, a non-local correlation if and only if
$$ \lim_{m,n\rightarrow\infty} \frac{1}{m}\|GH^t\|_{\ell_{\infty}^n\otimes_{\pi}\ell_{\infty}^n} = C >1. $$
The product of two Gaussian matrices is bi-orthogonally invariant and thanks to Theorem \ref{th:q-c-norms} we can lower bound its ${\ell_{\infty}^n\otimes_{\pi}\ell_{\infty}^n}$ norm. If $\alpha< \alpha_0$, the lower bound will tend to a number greater than $1$ with probability tending to $1$, as $m,n$ grow.

We will need the limiting empirical distribution of the product of two Gaussian matrices. This was first proven in \cite{Muller02}, and later generalized in several ways (see \cite{GKT14} and references therein).
\begin{theorem}\label{LawProdGauss}
Let $G=(g_{i,j})_{i, j=1}^{n,m_n}$ and $H=(h_{i,j})_{i, j=1}^{n,m_n}$ be two independent Gaussian random matrices, such that $\alpha:=\lim_{n\rightarrow \infty} m_n/n \in (0,\infty)$. For every $i=1,\ldots,n$, let $\lambda_i$ be the $i$-th eigenvalue of $GH^tHG^t/(nm_n)$, and define the empirical eigenvalue distribution of $GH^tHG^t/(nm_n)$ as
$$ F_n(x)=\frac 1 n \sum_{i=1}^n \chi_{\{\lambda_i\leq x\}} . $$
Then, almost surely,
$$ F_n(x)\underset{n\rightarrow\infty}{\longrightarrow} F(x), $$
where the Stieltjes transform of the distribution function $F(x)$, i.e.
$$s(z)=\int_{-\infty}^{\infty} \frac 1 {x-z} dF(x),$$
is determined by the equation
$$z s(z)-  {z s^2(z)}\left(\frac{\alpha -1 +z s(z)}\alpha\right)=1.$$
\end{theorem}

We will also use the well known bounds on the norm of a Gaussian vector (see for instance \cite{Barvinok}, Corollary 2.3).

\begin{proposition}\label{concentracion}
Let $(g_{i,j})_{i,j=1}^{n,m}$ be a $n\times m$ Gaussian matrix, and for every $i=1,\dots,n$, define $g_i=(g_{i,j})_{j=1}^m$, Gaussian vector in $\R^m$. Then, for every $i=1,\dots,n$ and every $0<\epsilon<1$,
$$\P\left(\|g_i\|\geq \frac{\sqrt{m}}{\sqrt{1-\epsilon}}\right)\leq e^{-\epsilon^2 m/4}\quad \text{and}\quad \P\left( \|g_i\| \leq \sqrt{m}\sqrt{1-\epsilon}\right)\leq e^{-\epsilon^2 m/4}.$$
As a consequence we have in particular that, for every $0<\epsilon<1$,
$$ \P\left(\max_{i=1,\dots, n}\left\|\frac{g_i}{\|g_i\|}-\frac{g_i}{\sqrt{m}}\right\|>\epsilon\right)\leq 2ne^{-\epsilon^2m/4}.\ $$
\end{proposition}

Here, we apply Theorem \ref{th:q-c-norms} to the random correlation matrices studied in \cite{GJPV} and described above.

\begin{corollary}\label{alphavalue} Let $G=(g_{i,j})_{i, j=1}^{n,m}$ and $H=(h_{i,j})_{i, j=1}^{n,m}$ be two independent Gaussian random matrices. Let $\alpha\in (0,+\infty)$ and $m=\alpha n$. Let $C_\alpha=\int_{-\infty}^{+\infty} t^{1/2} dF(t)$, where $F$ is, as in Theorem \ref{LawProdGauss}, the asymptotic empirical eigenvalue distribution of $GH^tHG^t/(nm)$.
Then, with probability $1-o(1)$,
\[\frac 1 {m} \| GH^t\|_{\ell_{\infty}^n\otimes_{\pi}\ell_{\infty}^n}\geq \left(\sqrt{\frac{16}{15}}- o(1)\right)\frac{C_\alpha}{\sqrt{\alpha}}.\]
Moreover, there exists an $\alpha_0$ such that, for any $\alpha < \alpha_0\approx 0.1269$, with probability $1-o(1)$,
\[\lim_{ n \rightarrow\infty} \frac 1 m \| GH^t\|_{\ell_{\infty}^n\otimes_{\pi}\ell_{\infty}^n}>1.\]
\end{corollary}

\begin{proof}
As $GH^t$ is bi-orthogonally invariant and with flat spectrum, according to equation (\ref{eq:q-c-norms}), we have that with high probability
\[\frac 1 m \|GH^t\|_{\ell_{\infty}^n\otimes_{\pi}\ell_{\infty}^n}  \geq \left(\sqrt{\frac{16}{15}} - o(1)\right)\frac 1 m \gamma_2(GH^t)\geq \left(\sqrt{\frac{16}{15}} - o(1)\right)\frac 1 {mn} \|GH^t\|_1,\]
where the second inequality follows from Proposition \ref{prop:lower-q}.
Applying Theorem \ref{LawProdGauss}, we obtain that with probability tending to $1$ the first claim holds:
\[\frac 1 m \|GH^t\|_{\ell_{\infty}^n\otimes_{\pi}\ell_{\infty}^n}\geq \left(\sqrt{\frac{16}{15}} - o(1)\right)\frac{1}{\sqrt{\alpha}n} \left\|\frac{GH^t}{\sqrt {n m}}\right\|_1 \geq \left(\sqrt{\frac{16}{15}}- o(1)\right)\frac{C_\alpha}{\sqrt{\alpha}}.\]
The existence of $\alpha_0$ follows from the continuity of the density function of the distribution $F$ as a function of $\alpha$. For an analytic expression of this density function see \cite{DP14}. Numerical evaluation of $C_\alpha$ gives the approximation of $\alpha_0$.
\end{proof}

\begin{proof}[Proof of Theorem \ref{prodgauss}]
We will see the unit vectors $u_1,\dots, u_n, v_1,\dots, v_n$ in $\R^m$ as arising from independent normalized Gaussian vectors. Let $G=(g_{i,j})_{i, j=1}^{n,m}$ and $H=(h_{i,j})_{i, j=1}^{n,m}$ be two random matrices whose entries are independent real standard Gaussian variables. For every $i,j=1,\dots , n$, let $g_i=(g_{i,k})_{k=1}^m$ and $h_j=(h_{j,k})_{k=1}^m$ be the row vectors of $G$ and $H$ respectively. Then, the vectors $u_i=g_i/\|g_i\|$ and $ v_j=h_j/\|h_j\|$ are independent uniformly distributed unit vectors in $\R^m$.

The matrix $\tau=(\langle u_i,v_j\rangle)_{i,j=1}^n$ is by construction a quantum correlation matrix. We will show that it is non-local, i.e. that $\|\tau\|_{\ell_{\infty}^n\otimes_{\pi}\ell_{\infty}^n}>1$, with probability tending to $1$ as $n$ goes to infinity. We can write
$$\tau=\frac 1 m GH^t + \left(\tau-\frac 1 m GH^t\right).$$
According to Corollary \ref{alphavalue} we know that $\lim_{ n \rightarrow\infty} \| GH^t\|_{\ell_{\infty}^n\otimes_{\pi}\ell_{\infty}^n}/m=c(\alpha)>1$ if $\alpha< \alpha_0$, with probability tending to $1$ as $n$ grows. To finish the proof we just need to prove that, with probability tending to $1$ as $n$ goes to infinity, $\lim_{ n \rightarrow\infty}\|\tau-GH^t/m \|_{\ell_{\infty}^n\otimes_{\pi}\ell_{\infty}^n}=0$ for $\alpha < \alpha_0$, so that
\[ \|\tau\|_{\ell_{\infty}^n\otimes_{\pi}\ell_{\infty}^n}\geq \left\|\tau-\frac 1 m GH^t \right\|_{\ell_{\infty}^n\otimes_{\pi}\ell_{\infty}^n}+\left\|\frac 1 m GH^t \right\|_{\ell_{\infty}^n\otimes_{\pi}\ell_{\infty}^n}\underset{n\rightarrow\infty}{\longrightarrow} c(\alpha)> 1.\]
In order to show this, we define, for every $i,j=1,\ldots,n$, the vectors $\varepsilon_i:= g_i/\|g_i\|- g_i/\sqrt{m}$ and $\delta_j:= h_j/\|h_j\|- h_j/\sqrt{m}$ in $\R^m$. Then we have
\[\tau_{i,j}- \frac 1 m \langle g_i,h_j\rangle= \frac{1}{\sqrt{m}}\langle \varepsilon_i,h_j\rangle+ \frac{1}{\sqrt{m}}\langle g_i,\delta_j\rangle + \langle \varepsilon_i,\delta_j\rangle. \]
Thus, by the triangle inequality, the quantity $\|\tau- GH^t/m\|_{\ell_{\infty}^n\otimes_{\pi}\ell_{\infty}^n}$ is upper bounded by
\[ \frac{1}{\sqrt{m}}\big\| (\langle \varepsilon_i,h_j\rangle)_{i,j=1}^{n} \big\|_{\ell_{\infty}^n\otimes_{\pi}\ell_{\infty}^n} + \frac{1}{\sqrt{m}} \big\| (\langle g_i,\delta_j\rangle)_{i,j=1}^{n} \big\|_{\ell_{\infty}^n\otimes_{\pi}\ell_{\infty}^n} +\big\| (\langle \varepsilon_i,\delta_j\rangle)_{i, j=1}^{n} \big\|_{\ell_{\infty}^n\otimes_{\pi}\ell_{\infty}^n}. \]
Hence, it is enough to show that each of the terms above tend to $0$ as $n$ goes to infinity. We just prove it for the first term, being the reasoning to prove it for the others entirely analogous.

Grothendieck's inequality (\ref{eq:groth}) allows us to bound the ${\ell_{\infty}^n\otimes_{\pi}\ell_{\infty}^n}$ norm of $\left(\langle \varepsilon_i,h_j\rangle\right)_{i,j=1}^{n}$ in terms of its $\gamma_2$ norm, namely
\[\left\|\left(\langle \varepsilon_i,h_j\rangle\right)_{i, j=1}^{n}\right\|_{\ell_{\infty}^n\otimes_{\pi}\ell_{\infty}^n}\leq K_G\, \gamma_2\left(\left(\langle \varepsilon_i,h_j\rangle\right)_{i, j=1}^{n}\right)\leq K_G\, \max_{i=1,\ldots,n} \|\varepsilon_i\|_2 \max_{j=1,\ldots,n} \left\|h_j\right\|_2,\]
where the last inequality follows from the definition of the $\gamma_2$ norm, recalled as equation (\ref{def:gamma_2}). Now, by Proposition \ref{concentracion} and a union bound argument, we know that, for any $\epsilon>0$,
\[ \P\left(\max_{i=1,\ldots,n} \|\varepsilon_i\|_2>\epsilon\right)\leq 2ne^{-\epsilon^2m/4}\ \text{and}\ \P\left(\max_{j=1,\ldots,n} \left\|h_j\right\|_2>\frac{\sqrt{m}}{\sqrt{1-\epsilon}}\right)\leq ne^{-\epsilon^2m/4}. \]
So putting everything together, we get that, for any $\epsilon>0$,
\[\P\left(\frac{1}{\sqrt{m}}\left\|\left(\langle \varepsilon_i,h_j\rangle\right)_{i,j=1}^{n}\right\|_{\ell_{\infty}^n\otimes_{\pi}\ell_{\infty}^n} \leq K_G \,\frac \epsilon {\sqrt{1-\epsilon}}\right) \geq 1-3ne^{-\epsilon^2m/4}.\]

This completes the proof of Theorem \ref{prodgauss}.
\end{proof}

\

\section*{Acknowledgements}
This research was financially supported by the European Research Council (advanced grant IRQUAT ERC-2010-AdG-267386), the Spanish MINECO (projects FIS2013-40627-P, MTM2014-54240-P, ICMAT Severo Ochoa SEV2015-0554 and ``Ram\'on y Cajal'' program), the Comunidad de Madrid (QUITEMAD+ project S2013-ICE-2801), the Generalitat de Catalunya (CIRIT project 2014-SGR-966), the French CNRS (ANR project Stoq 14-CE25-0033), and John Templeton Foundation grant 48322. The opinions expressed in this publication are those of the authors and do not necessarily reflect the views of the John Templeton Foundation.

\

\end{document}